\documentclass[a4paper,envcountsame,10pt]{article}
\usepackage{authblk}
\title{Encoding Turing Machines\\ into the Deterministic $\l$-Calculus}
\author[1]{Ugo Dal Lago}
\affil[1]{Universit\'a di Bologna \& INRIA Sophia Antipolis, Italy}

\author[2]{Beniamino Accattoli}
\affil[2]{INRIA, UMR 7161, LIX, \'Ecole Polytechnique, France}

\usepackage{ifthen}

\newboolean{talk}
\setboolean{talk}{false}
\newboolean{paper}
\setboolean{paper}{false}

\newboolean{IEEEstyle}
\setboolean{IEEEstyle}{false}
\newboolean{lipicsstyle}
\setboolean{lipicsstyle}{false}
\newboolean{eptcsstyle}
\setboolean{eptcsstyle}{false}
\newboolean{entcsstyle}
\setboolean{entcsstyle}{false}
\newboolean{sigplanstyle}
\setboolean{sigplanstyle}{false}
\newboolean{easychairstyle}
\setboolean{easychairstyle}{false}

\newboolean{needstheorems}
\setboolean{needstheorems}{false}
\newboolean{withimages}
\setboolean{withimages}{false}
\newboolean{withproofs}
\setboolean{withproofs}{true}

\newboolean{french}
\setboolean{french}{false}

\setboolean{needstheorems}{true}
\ifthenelse{\boolean{talk}}{}{\usepackage[utf8]{inputenc}}
\ifthenelse{\boolean{french}}{\usepackage[french]{babel}}{
  \ifthenelse{\boolean{lipicsstyle}}{}{\usepackage[english]{babel}}}
\usepackage{amssymb}
\usepackage{amsmath}
\usepackage{graphicx}
\usepackage{bussproofs}
\usepackage{cmll}
\usepackage{stmaryrd}
\usepackage{multirow}
\ifthenelse{\boolean{talk}}{\usepackage{color}}{
  \ifthenelse{\boolean{lipicsstyle}}{\usepackage{color}}{\usepackage[usenames]{xcolor}}}
\ifthenelse{\boolean{IEEEstyle}}{
\usepackage[bookmarks={false}]{hyperref}}{
\usepackage{hyperref}}
\usepackage{fancybox}
\usepackage{xspace}
\usepackage{hhline}

\ifthenelse{\boolean{IEEEstyle}}{
\setboolean{needstheorems}{true}}{}

\ifthenelse{\boolean{eptcsstyle}}{
\setboolean{needstheorems}{true}}{}

\ifthenelse{\boolean{sigplanstyle}}{
\setboolean{needstheorems}{true}}{}

\ifthenelse{\boolean{easychairstyle}}{
\setboolean{needstheorems}{true}}{}

\ifthenelse{\boolean{lipicsstyle}}{
    \newtheorem{proposition}[theorem]{Proposition}
    }{}

\ifthenelse{\boolean{needstheorems}}{

\usepackage{amsthm}

    \newtheorem{theorem}{Theorem}[section]
    \newtheorem{lemma}[theorem]{Lemma}
    \newtheorem{corollary}[theorem]{Corollary}

}{}

\newcommand{\myproof}[1]{
\ifthenelse{\boolean{withproofs}}{#1}{}
}

\newcommand{\la}[1]{\lambda #1.}

\newcommand{\tm}{t}
\newcommand{\tmtwo}{s}
\newcommand{\tmthree}{u}
\newcommand{\tmfour}{r}

\newcommand{\var}{x}
\newcommand{\vartwo}{y}
\newcommand{\varthree}{z}

\newcommand{\rootRew}[1]{\mapsto_{#1}}
\newcommand{\Rew}[1]{\rightarrow_{#1}}

\newcommand{\tob}{\Rew{\beta}}
\newcommand{\rtob}{\rootRew{\beta}}

\newcommand{\val}{v}
\newcommand{\valtwo}{\val'}
\newcommand{\valthree}{\val''}

\newcommand{\ctxholep}[1]{\langle #1\rangle}
\newcommand{\ctxhole}{\ctxholep{\cdot}}

\newcommand{\evctx}{E}

\newcommand{\evctxp}[1]{\evctx\ctxholep{#1}}

\newcommand{\nbvctxtwo}[1]{\nbvctxtwo{#1}}

\newcommand{\defeq}{:=}
\newcommand{\grameq}{::=}

\newcommand{\isub}[2]{\{#1/#2\}}

\renewcommand{\isub}[2]{\{#1{\shortleftarrow}#2\}}

\newcommand{\llbrace}{\{ \kern -0.27em \vert}
\newcommand{\rrbrace}{\vert \kern -0.27em \}}

\renewcommand{\l}{\lambda}
\newcommand{\ie}{{\em i.e.}\xspace}

\newcommand{\ih}{{\textit{i.h.}}\xspace}

\newcommand{\ignore}[1]{}

\newcommand{\myinput}[1]{\ifthenelse{\boolean{withimages}}{\input{#1}}{}}

\newcommand{\reflemma}[1]{Lemma~\ref{l:#1}}

\newcommand{\reflemmaeq}[1]{{L.\ref{l:#1}}}

\ifthenelse{\boolean{easychairstyle}}{}{
	
}

\newcommand{\refcoroeq}[1]{Cor.~\ref{coro:#1}}

\newcommand{\set}[1]{\{#1\}}

\newcommand{\size}[1]{|#1|}

\newcommand{\state}{s}

\newcommand{\tomachhole}[1]{\leadsto_{#1}}

\newcommand{\withproofs}[1]{\ifthenelse{\boolean{withproofs}}{#1}{}}

\newcommand{\withoutproofs}[1]{\ifthenelse{\boolean{withproofs}}{}{#1}}

 \ifthenelse{\boolean{entcsstyle}}{}{
 
 }

\newcommand{\detLam}{\Lambda_{\tt det}}
\newcommand{\tobdet}{\rightarrow_{det}}
\newcommand{\TM}{{\cal M}}
\newcommand{\alpone}{\Sigma}
\newcommand{\cods}[1]{\overline{#1}}

\newcounter{numberone}

\newcommand{\fix}{\theta}
\newcommand{\Nset}{\mathbb{N}}
\newcommand{\M}{\TM}

\newcommand{\elone}{a}
\newcommand{\elem}{a}
\newcommand{\cod}[2]{\overline{#1}^{#2}}

\newcommand{\alptwo}{\Delta}

\newcommand{\strone}{s}
\newcommand{\strtwo}{r}
\newcommand{\strthree}{p}
\newcommand{\append}{{\tt{append}}}
\newcommand{\appendalpchar}[2]{\append_{#1}^{#2}}
\newcommand{\appendalph}[1]{\append_{#1}}
\newcommand{\appendalphs}{\append}
\newcommand{\appendchar}[1]{\append^{#1}}

\newcommand{\lift}{{\tt lift}}
\newcommand{\liftalp}{\lift_{\alpone}}
\newcommand{\liftaux}{{\tt liftaux}}
\newcommand{\liftauxalp}{\liftaux_{\alpone}}

\newcommand{\flatalp}{{\tt flat}_{\alponeblank}}
\newcommand{\flataux}{{\tt flataux}}
\newcommand{\flatauxalp}{\flataux_{\alponeblank}}

\newcommand{\alponeblank}{\alpone_{\elemblank}}

\newcommand{\funconvert}[1]{f_{#1}}

\newcommand{\init}[2]{{\tt init}^{#1}_{#2}}
\newcommand{\inits}{{\tt init}}

\newcommand{\elemblank}{\Box}
\renewcommand{\state}{q}
\newcommand{\States}{Q}
\newcommand{\statein}{\state_{\mathit{in}}}
\newcommand{\statefin}{\state_{\mathit{fin}}}

\newcommand{\config}{C}
\newcommand{\configtwo}{D}
\newcommand{\final}[2]{{\tt final}^{#1}_{#2}}
\newcommand{\finals}{{\tt final}}

\newcommand{\transaux}{{\tt transaux}}
\newcommand{\trans}[1]{{\tt trans}^{#1}}
\newcommand{\transs}{{\tt trans}}

\newcommand{\function}{\overline{\M}}

\newcommand{\tomachtur}{\tomachhole{\M}}

\newcommand{\initconfig}{\config_{\tt in}^\M(\strone)}
\newcommand{\initconfigs}{\config_{\tt in}(\strone)}
\newcommand{\finalconfig}{\config_{\tt fin}^\M(\strone)}
\newcommand{\finalconfigs}{\config_{\tt fin}(\strone)}
\newcommand{\cont}{k}

\newcommand{\tuple}[1]{\langle #1 \rangle}
\date{}
\begin{document}
\maketitle



\section{Introduction}

This note is about encoding Turing machines into the $\l$-calculus. The encoding we show is interesting for two reasons:
\begin{enumerate}
  \item \emph{Weakly strategy independent}: the image of the encoding is a very small fragment of the $\l$-calculus, that we call \emph{the deterministic $\l$-calculus} $\detLam$. Essentially, it is the CPS (continuation-passing style) $\l$-calculus restricted to weak evaluation (\ie, not under abstractions). In $\detLam$ every term has at most one redex, and so all weak strategies collapse into a single deterministic evaluation strategy, because there are no choices between redexes to be made. The important consequence of this property is that every weak evaluation strategy then allows to simulate Turing machines, as well as any strong strategy reducing weak head redexes (or even only weak head redexes) first.
  \item \emph{Linear overhead}: the simulation is very efficient, when taking the number of $\beta$-steps as the time cost model for the deterministic $\l$-calculus. The simulation in $\detLam$ indeed requires a number of $\beta$-steps that is linear in the number of transitions of the encoded Turing machine, which is the best possible overhead. Therefore, not only all weak strategies simulate Turing machines, but they all do it \emph{efficiently}.
\end{enumerate}

The encoding has been conceived by Ugo Dal Lago in order to provide a simulation of Turing machines by head evaluation, that was part of the results of Accattoli and Dal Lago's \cite{DBLP:conf/rta/AccattoliL12}. The encoding itself appeared only in the technical report associated to \cite{DBLP:conf/rta/AccattoliL12} and it did not receive much attention.

This note is due to Beniamino Accattoli. Essentially, it is a commented and smoothed presentation of Dal Lago's development whose aims are:
\begin{enumerate}
\item to stress the relevance of the weak-strategy-independent property of the encoding,
\item to make the proof accessible to a wider public.
\end{enumerate}
The results of this note, and their consequences, are discussed and put in context in Accattoli's paper \emph{(In)Efficiency and Reasonable Cost Models} \cite{LSFA2017invited}.

\section{Deterministic $\l$-Calculus}
The language and the evaluation contexts of the deterministic $\l$-calculus $\detLam$ are given by:
\begin{center}
$\begin{array}{r@{\hspace{.5cm}}rlllllll}
    \mbox{Terms} & \tm,\tmtwo,\tmthree,\tmfour &\grameq& \val \mid  \tm \val\\
    \mbox{Values} & \val, \valtwo, \valthree & \grameq & \la\var\tm \mid \var\\\\
   \mbox{Evaluation Contexts} & \evctx  & \grameq & \ctxhole\mid  \evctx\val
\end{array}$
\end{center}

Note that 
\begin{itemize}
\item \emph{Arguments are values}: the right subterm of an application has to be a value, in contrast to what happens in the ordinary $\l$-calculus. 
\item \emph{Weak evaluation}: evaluation contexts are \emph{weak}, \ie they do not enter inside abstractions. 
\end{itemize}

Evaluation is then defined by:
\begin{center}
$\begin{array}{c@{\hspace{.5cm}}rll}
      \textsc{Rule at top level} & \multicolumn{3}{c}{\textsc{Contextual closure}} \\
	    \!(\la\var\tm)\tmtwo \rtob \tm \isub\var\tmtwo &
	    \multicolumn{3}{c}{\evctxp \tm \tobdet \evctxp \tmtwo \textrm{~~~if } \tm \rtob \tmtwo} \\
\end{array}$
\end{center}

\emph{Convention}: to improve readability we omit some parenthesis, giving precedence to application with respect to abstraction. Therefore $\la \var \tm \tmtwo$ stands for $\la \var (\tm \tmtwo)$ and not for $(\la \var \tm) \tmtwo$, that instead requires parenthesis.

The name of this calculus is motivated by the following immediate lemma.

\begin{lemma}
Let $\tm \in \detLam$. There is at most one $\tmtwo \in \detLam$ such that $\tm \tob \tmtwo$, and in that case $\tm$ is an application.
\end{lemma}

\begin{proof}
By induction on $\tm$. If $\tm$ is a value then it does not reduce. Then assume that $\tm$ is an application $\tm = \tmtwo \val$. Let's apply the \ih to $\tmtwo$. Two cases:
\begin{enumerate}
\item \emph{$\tmtwo$ reduces and it is an application}: then $\tm$ has one redex, the one given by $\tmtwo$ (because $\ctxhole \val$ is an evaluation context), and no other one, because $\val$ does not reduce and $\tmtwo$ is not an abstraction by the \ih

\item \emph{$\tmtwo$ does not reduce}: if $\tmtwo$ is not an abstraction then $\tm$ is normal, otherwise $\tmtwo = \la\var\tmthree$ and $\tm = (\la\var\tmthree) \val$ has exactly one redex.
\end{enumerate}
\end{proof}

\paragraph{Fixpoint operator.} The encoding of Turing machines requires a general form of recursion, that is usually implemented via a fixpoint combinator. We use Turing's fixpoint combinator, in its call-by-value variant, that fits into $\detLam$ and that returns a fixpoint up to $\eta$-equivalence. It is defined as follows: let $\fix$ the term $\tm\tm$, where
$
\tm \defeq \la\var \la\vartwo \vartwo(\la\varthree \var \var \vartwo \varthree).
$
Now, given a term $\tmtwo$ let us show that $\fix\tmtwo$ is a fixpoint of $\tmtwo$ up to $\eta$-equivalence.
$$\begin{array}{lcllll}
\fix\tmtwo  
& = &
(\la\var \la\vartwo \vartwo(\la\varthree \var \var \vartwo \varthree)) \tm \tmtwo \\
& \tobdet & 
(\la\vartwo \vartwo(\la\varthree \tm \tm \vartwo \varthree)) \tmtwo \\

& \tobdet & 
\tmtwo (\la\varthree \tm \tm \tmtwo \varthree)\\ 
& = &
\tmtwo (\la\varthree \fix \tmtwo \varthree)\\
& =_\eta &
\tmtwo (\fix \tmtwo)
\end{array}$$

It is well-known that $\eta$-equivalent terms are indistinguishable in the $\l$-calculus (this is Bohm theorem). Therefore, we will simply use the fact that $\fix\tmtwo \tobdet^2 \tmtwo (\la\varthree \fix \tmtwo \varthree)$ wihtout dealing with $\eta$-equivalence. This fact will not induce any complication.

Let us a point out a subtlety. The evaluation of the fixpoint operator $\fix$ terminates because evaluation is weak, indeed $\fix \tobdet \la\vartwo \vartwo(\la\varthree \fix \vartwo \varthree)$ and the new occurrence of $\fix$ is under abstraction, so that the reduct is a $\detLam$ normal form. This fact explains why the encoding is \emph{not} strategy independent with respect to strong strategies: the fixpoint operator is used to implement the repeated application of the transition function of a Turing machine, then a (strong) perpetual strategy always diverges on the image of our encoding.

Last, in the encodings of the next sections the only point where duplication will take place will be in the use of the fixpoint combinator, all other constructions will be linear (affine, to be precise).

  
\section{Overview of the Encoding}

The encoding of a Turing machine $\M$ into the Deterministic $\l$-Calculus $\detLam$ is built in a few steps. There are various points worth stressing, some of them further organized in sub-points.

\begin{enumerate}
  \item \emph{Complexity parameters}: the complexity of the encoding will be measured with respect to the size $\size\strone$ of the input string $\strone$ on which $\M$ is executed. The cost model considered on $\detLam$ is the number of $\beta$-steps. In general, the number of $\beta$-steps taken by the simulation will also depend on the size of the alphabet $\alpone$ and the size of the set of states $\States$ of $\M$. These quantities, however, will be considered as fixed parameters of the problem, that is, as constants.

  \item \emph{Encoding and converting strings}:
  
  \begin{enumerate}    
	\item \emph{Alphabets and strings}: The encoding relies on Scott's encoding of alphabets and strings, that fits into $\detLam$. Strings will be used to encode the input, the output, and the tape. The encoding of a string $\strone$ will depend (linearly) on the cardinality of the underlying alphabet $\alpone$, so that encodings of the same string with respect to different alphabets will be different. Note, however, that the size $\size\alpone$ of the alphabet is considered as a fixed parameter, so that the encoding of an element of $\alpone$ takes constant space in our analysis. The set of states of a Turing machine will also be looked at as an alphabet.
	
	\item \emph{Appending on strings}: the encoding repeatedly uses a term $\appendalpchar{\alpone}{\elem}$ that takes as input the encoding $\cod\strone{\alpone^*}$ of a string with respect to $\alpone$ and appends the encoding of the element $\elem$ of $\alpone$, returning the encoding $\cod{(\elem\strone)}{\alpone^*}$ of the compound string $\elem\strone$. This will be needed for various reasons, but mainly to simulate the movement of the machine head, since we will have to append a character to the string on the left or on the right of the head. The appending operation will take constant time.
	
	\item \emph{Converting encodings of a string with respect to different alphabets}: A turing machine $\M$ receives and produces strings in a given alphabet $\alpone$ but $\M$ actually works on the alphabet $\alpone$ extended with the blank symbol $\elemblank$. Since the encoding of strings depends on the alphabet we have to show that there are terms $\liftalp$ and $\flatalp$ that can efficiently switch between $\cod\strone{\alpone^*}$ and $\cod\strone{(\alpone\uplus\set\elemblank)^*}$. The conversion operations will be linear in the size $\size\strone$ of the converted string $\strone$, and of course they will rely on the appending operation of the previous point.
  \end{enumerate}
  
  \item \emph{Encoding Turing machines}: 
  \begin{enumerate} 
  \item \emph{Encoding configurations}: a configuration $\config$ of $\M$ is encoded as a tuple $\cod\config\M$ of characters and strings depending on the alphabet $\alpone$ and the set of states $\States$ of $\M$. Actually, the strings of $\config$ appear in $\cod\config\M$ encoded with respect to the extended alphabet $\alpone\uplus\set\elemblank$. The encoding $\cod\config\M$ of a configuration $\config$ requires space linear in $\size\config$ (again, it is also linear in $\size\alpone$ and $\size\States$ but these quantities are considered as constants).
  
  \item \emph{Turning the input string into the initial configuration}:
    we will build a term $\initconfig$ that takes in input the encoding $\cod\strone{\alpone^*}$ of the input string $\strone$ and builds the initial configuration $\initconfig$ of $\M$ with respect to $\strone$. Note that $\init{\M}{}$ has to convert the representation of $\strone$ from the one with respect to $\alpone$ to the one with respect to $\alpone\uplus\set\elemblank$, and thus will use the conversion operation $\liftalp$. Building the initial configuration $\initconfig$ from $\cod\strone{\alpone^*}$ will take time linear in $\size\strone$. 
    
  \item \emph{Simulation of a machine transition}:
    we will build a term $\trans{\M}$ that encodes the transition function $\delta$ of $\M$ and that it will take a constant number of $\beta$-steps to simulate every single transition. Up to the many technicalities of the encoding, $\trans{\M}$ simply encodes $\delta$ as a table. To show that $\trans{\M}$ properly simulates $\delta$ is the most involved point of the proof.
    
  \item\emph{Extracting the output from the final configuration}:
    the simulation of the transition function produces an encoding of the final configuration $\finalconfig$ with respect to an output string $\strone$. Note that, as for the initial configuration, in $\finalconfig$ the string $\strone$ is encoded with respect to $\alpone\uplus\set\elemblank$ while the output is expected to be encoded with respect to $\alpone$. We then apply a term $\final{\M}{}$ that extracts $\cod\strone{ (\alpone\uplus\set\elemblank)^* }$ from $\finalconfig$ and then converts it to $\cod\strone{\alpone^*}$ by means of the flattening operations $\flatalp$. Such an extracting operation will take time linear in $\size\strone$.    
\end{enumerate}

  \item \emph{Continuation passing style}: to stick to $\detLam$ the encoding will actually be in continuation-passing style. Concretely, the terms  $\appendalpchar{\alpone}\elem$, $\liftalp$, $\flatalp$, $\init{\M}{}$, $\trans{\M}$, and $\final{\M}{}$ will all take as first argument a continuation term $\cont$. 
  
  \item \emph{The image of the encoding.}  The encoding will produce a closed $\l$-term $\tm \in \detLam$. The fact that it is closed is natural but it does not play any role. The fact that it falls into the deterministic $\l$-calculus instead is essential, not for the proof itself, but for the significance of the result for the study of reasonable cost models, explained in \cite{LSFA2017invited}.
\end{enumerate}

\section{Encoding and Converting Strings}

\paragraph{Encoding alphabets.} Let $\alpone=\{\elone_1,\ldots,\elone_n\}$ be a finite alphabet. Elements of $\alpone$ are encoded
as follows:
$$
\cod{\elone_i}{\alpone} \defeq \lambda\var_1.\ldots.\lambda\var_n.\var_i \;.
$$
When the alphabet will be clear from the context we will simply write $\cods\elone_i$. Note that 
\begin{enumerate}
\item the representation fixes a total order
on $\alpone$ such that $\elem_i\leq \elem_j$ iff $i\leq j$;
\item the representation of an element $\cod{\elone_i}{\alpone}$ requires space linear (and not logarithmic) in $\size\alpone$. But, since $\alpone$ is fixed, it actually requires constant space.
\end{enumerate}

\paragraph{Encoding strings.} A string in
$\strone\in\alpone^*$ is represented by a term
$\cod\strone{\alpone^*}$, defined by induction on the structure of $\strone$ as follows:
\begin{align*}
\cod{\varepsilon}{\alpone^*} & \defeq \lambda\var_1.\ldots.\lambda\var_n.\lambda\vartwo.\vartwo\;,\\
\cod{\elone_i\strtwo}{\alpone^*} & \defeq \lambda\var_1.\ldots.
   \lambda\var_n.\lambda\vartwo.\var_i\cod{\strtwo}{\alpone^*}.
\end{align*}
Note that the representation depends on the cardinality
of $\alpone$. In other words, if $\strone\in\alpone^*$
and $\alpone\subset\alptwo$, $\cod\strone{\alpone^*}\neq\cod\strone{\alptwo^*}$. In particular, $\size{\cod\strone{\alpone^*}} = \Theta(\size\strone\cdot\size\alpone)$. The size of the alphabet is however considered as a fixed parameter, and so we rather have $\size{\cod\strone{\alpone^*}} = \Theta(\size\strone)$.

\paragraph{Appending a character at the beginning of a string.} We now show that, for any given alphabet $\alpone$, there is a term $\appendalph{\alpone}$ that appends a character at the beginning of a string. The number of steps taken by $\appendalph{\alpone}$ to produce the compound string is linear in $\size\alpone$, and thus it is constant. Moreover, it takes a continuation $\cont$ as first parameter. Last, since here there is a single alphabet in use, we write $\appendalphs$, $\cods\elem$, and $\cods\strone$ rather than $\appendalph{\alpone}$, $\cod{\elem}{\alpone}$, and $\cod{\elem}{\alpone^*}$.

\begin{lemma}[Appending a character in constant time]
\label{l:append-char}
Let $\alpone$ be an alphabet and $\elem \in\alpone$ one of its characters. There is 
a term $\appendalpchar\alpone\elem$ 
such that for every
continuation $\cont$ and every string $\strone\in\alpone^*$, 
$$
\appendalpchar\alpone\elem \cont  \cods\strone
\tobdet^{O(1)} 
\cont \cods{(\elem\strone)}.
$$
\end{lemma}

\begin{proof}
Define the term $\appendalpchar \alpone \elem \defeq \la\strone \la\cont \cont(\lambda \var_1.\ldots.\lambda \var_{\size\alpone}.\lambda \vartwo.\var_{i_\elem} \strone)$ where $i_\elem$ is the index of $\elem$ in the ordering of $\alpone$ fixed by its encoding, that appends the character $\elem$ to the string $\strone$ relatively to the alphabet $\alpone$. We have:
$$\begin{array}{rclcl}
  \appendalpchar \alpone \elem \cods\strone \cont
  & = &
  (\la\strone \la\cont \cont(\lambda \var_1.\ldots.\lambda \var_{\size\alpone}.\lambda \vartwo.\var_{i_\elem} \strone)) \cods\strone \cont \\
  & \tobdet^2 &
  \cont(\lambda \var_1.\ldots.\lambda \var_{\size\alpone}.\lambda \vartwo.\var_{i_\elem}\cods{\strone}) \\
  & = &
  \cont \cods{(\elem \strone)}.
\end{array}$$
\end{proof}

\paragraph{Lifting the encoding of a string to the alphabet extended with the blank symbol.}
We will have to convert the encoding $\cod\strone\alpone$ to the encoding $\cod\strone{\alpone \uplus\set\elemblank}$ with respect to the same alphabet $\alpone$ but extended with the blank symbol $\elemblank$, and viceversa. We use the notation $\alponeblank \defeq \alpone \cup \set\elemblank$ and we use the convention that the encoding with respect to $\alponeblank$ uses the same order on the elements of $\alpone$ extended by adding the blank symbol $\elemblank$ as the last element.

The underlying idea is very simple, and as expected: the encoding of every character $\cod\elem\alpone$ is sent to $\cod\elem\alponeblank$ and appended to the portion of the string already lifted so far, iterating over the string. Concretely, however, the transformation looks quite technical because the iteration is implemented via the fixpoint operator.

\begin{lemma}[Lifting the string encoding]
\label{l:lifting-alphabets}
Let $\alpone$ be a finite alphabet. There is a term $\liftalp$ such that 
$$
\liftalp \cont \cod\strone{\alpone^*}
\tobdet^{\theta(\size\strone)}
\cont \cod{\strone}{\alponeblank^*}.
$$
\end{lemma}

\begin{proof}
Let:
$$
\liftauxalp \defeq \la\var \la\cont \lambda \strone.\strone N_1\ldots N_{\size\alpone}N\cont,
$$
where $N\defeq\la\vartwo \vartwo\cod{\varepsilon}{\alponeblank^*}$ and $
N_i\defeq
    \la \strtwo \la \cont \var (\la\strthree\appendalpchar{\alponeblank}{\elem_i} \cont \strthree)\strtwo
$
 for any $i$. Finally, define $\liftalp   \defeq \fix \liftauxalp$.

We prove a more precise statement. By \reflemmaeq{append-char} there is a constant $c$ such that $
\appendalph\alponeblank \cont \cod\elem\alponeblank \cod\strone\alponeblank
\tobdet^c 
\cont \cod{\elem\strone}\alponeblank$. Now, we prove that 
$$
\liftalp \cont \cod\strone{\alpone^*}
\tobdet^{n}
\cont \cod{\strone}{\alponeblank^*}.
$$
where $\size\strone \leq n \leq \funconvert{\alpone}(\size\strone) \defeq \size\strone(\size\alpone + c + 10)+\size\alpone+6$. The fact that $\size\strone \leq n$ will be evident, so that we rather prove the other inequality. By induction on $\size\strone$.

Define the abbreviations $P_i \defeq N_i\isub\var {\la\varthree\liftalp \varthree}$. 
$$\begin{array}{rclcl}
  \liftalp \cont \cod\strone{\alptwo^*}
  & = &
  \fix \liftauxalp \cont \cod\strone{\alptwo^*}\\
  & \tobdet^2 &
  \liftauxalp (\la\varthree \fix \liftauxalp \varthree) \cont \cod\strone{\alptwo^*}\\
  & = &
  \liftauxalp (\la\varthree\liftalp \varthree) \cont \cod\strone{\alptwo^*}\\
  & = &
  (\la\var \la\cont \lambda \strone.\strone N_1\ldots N_{\size\alpone}N\cont) (\la\varthree\liftalp \varthree) \cont \cod\strone{\alpone^*}\\
  & \tobdet &
  (\la\cont \lambda \strone.\strone N_1\ldots N_{\size\alpone}N\cont) \isub\var{ (\la\varthree\liftalp \varthree) } \cont \cod\strone{\alpone^*}\\
  & = &
  (\lambda \cont.\lambda \strone.\strone P_1\ldots P_{\size\alpone}N\cont)\cont\cod\strone{\alpone^*}\\
  & \tobdet^2 & 
  \cod\strone{\alpone^*}P_1\ldots P_{\size\alpone}N\cont
\end{array}$$
Two cases, depending on the nature of $\strone$.
\begin{itemize}
\item \emph{$\strone$ is the empty string $\varepsilon$}. Then:
\begin{eqnarray*}
\cod{\varepsilon}{\alpone^*}P_1\ldots P_{\size\alpone}N\cont
   & =& (\lambda\var_1.\ldots.\lambda\var_{\size\alpone}.I) P_1\ldots P_{\size\alpone}N\cont\\
   &\tobdet^{\size\alpone}& I N\cont\\
      &\tobdet& N\cont\\
   & = & (\la\vartwo \vartwo\cod{\varepsilon}{\alponeblank^*}) \cont \\
   &\tobdet &\cont\cod{\varepsilon}{\alponeblank^*}
\end{eqnarray*}

So we have $\liftalp \cont\cod\strone{\alpone^*} 
\tobdet^{n}
\cont\cod{\varepsilon}{\alponeblank^*}$ with $\size\strone = 0 \leq n = \size\alpone +6 = \funconvert{\alpone}(0) = \funconvert{\alpone}(\size\strone)$.

\item \emph{$\strone$ is a compound string $\elem_i\strtwo$}. Then:
 \begin{eqnarray*}
  \cod{\elem_i\strtwo}{\alpone^*}P_1\ldots P_{\size\alpone}N\cont
   &=& (\lambda\var_1.\ldots.\lambda\var_{\size\alpone}.\lambda\vartwo.\var_i\cod{\strtwo}{\alpone^*}) P_1\ldots P_{\size\alpone}N\cont\\
   &\tobdet^{\size\alpone}& 
   (\lambda\vartwo.P_i\cod{\strtwo}{\alpone^*}) N\cont\\
   &\tobdet& 
   P_i\cod{\strtwo}{\alpone^*}\cont
\end{eqnarray*}
by \ih there is $m$ such that $\size\strtwo \leq m \leq \funconvert{\alpone}(\size\strtwo)$ and
\begin{eqnarray*}
	P_i\cod{\strtwo}{\alpone^*}\cont 
	& = & 
	(\la \strtwo \la \cont (\la\varthree\liftalp \varthree) (\la\strthree\appendalpchar{\alponeblank}{\elem_i} \cont \strthree)\strtwo) \cod{\strtwo}{\alpone^*}\cont\\
	& \tobdet^3 &
	\liftalp (\la\strthree\appendalpchar{\alponeblank}{\elem_i} \cont \strthree)\cod{\strtwo}{\alpone^*}\\
        (\mbox{by \ih}) & \tobdet^m &
        (\la\strthree\appendalpchar{\alponeblank}{\elem_i} \cont \strthree)\cod{\strtwo}{\alponeblank^*}\\
        & \tobdet &
        \appendalpchar{\alponeblank}{\elem_i} \cont\cod{\strtwo}{\alponeblank^*}\\
        (\mbox{by \reflemmaeq{append-char}}) & \tobdet^c &
        \cont\cod{(\elem_i\strtwo)}{\alponeblank^*}\\
        & = &
        \cont\cod{\strone}{\alponeblank^*}
\end{eqnarray*}

Summing up, $\liftalp \cont\cod\strone{\alpone^*}
\tobdet^{n}
\cont\cod{\strone}{\alponeblank^*}$ with $n = m + c + \size\alpone + 10$. Note that $\funconvert{\alpone}(\size\strone) = \funconvert{\alpone}(\size\strtwo+1) = \funconvert{\alpone}(\size\strtwo) + \size\alpone + c + 10$ and since $\size\strtwo \leq m \leq \funconvert{\alpone}(\size\strtwo)$ we obtain $\size\strone = \size\strtwo +1 \leq n \leq \funconvert{\alpone}(\size\strone)$.
\end{itemize}
\end{proof}

\paragraph{Flattening the encoding of a string to the alphabet without the blank symbol.} The following lemma shows that also the opposite operation on strings, namely converting the encoding of a string with respect to $\alponeblank$ to the enconding with respect to $\alpone$, is implementable in $\detLam$ in constant time. The operation is very similar to the lifting one, the proof only has an additional subcase for $\elemblank$.

\begin{lemma}[Flattening the string encoding]
\label{l:flattening-alphabets}
Let $\alpone$ be a finite alphabet. There is a term $\flatalp$ such that 
$$
\flatalp \cont \cod\strone{\alponeblank^*}
\tobdet^{\theta(\size\strone)}
\cont \cod{\strone}{\alpone^*}.
$$
\end{lemma}

\begin{proof}
Let:
$$
\flatauxalp \defeq \la\var \la\cont \lambda \strone.\strone N_1\ldots N_{\size\alpone}N_{\elemblank}N\cont,
$$
where 
\begin{itemize}
\item $N\defeq\la\vartwo \vartwo\cod{\varepsilon}{\alpone^*}$
\item $N_i\defeq
    \la \strtwo \la \cont \var (\la\strthree\appendalpchar{\alpone}{\elem_i} \cont \strthree)\strtwo
$
 for any $i$
 \item     $N_{\elemblank} \defeq \la\strtwo \la\cont \var \cont\strtwo$
 \end{itemize}
Finally, define $\flatalp   \defeq \fix \flatauxalp$.

We prove a more precise statement. By \reflemmaeq{append-char} there is a constant $c$ such that $
\appendalph\alpone \cont \cod\elem\alpone \cod\strone\alpone
\tobdet^c 
\cont \cod{\elem\strone}\alpone$. Now, we prove that 
$$
\flatalp \cont \cod\strone{\alponeblank^*}
\tobdet^{n}
\cont \cod{\strone}{\alpone^*}.
$$
where $\size\strone \leq n \leq \funconvert{\alponeblank}(\size\strone) \defeq \size\strone(\size\alponeblank + c + 10)+\size\alponeblank+6$. The fact that $\size\strone \leq n$ will be evident, so that we rather prove the other inequality. By induction on $\size\strone$. The structure of the proof is exactly as for the lifting case, but for the last subcase handling the blank element $\elemblank$.

Define the abbreviations $P_i \defeq N_i\isub\var {\la\varthree\flatalp \varthree}$ and $P_\elemblank \defeq N_\elemblank\isub\var {\la\varthree\flatalp \varthree} = \la\strtwo \la\cont (\la\varthree\flatalp \varthree) \cont\strtwo$. 
$$\begin{array}{rclcl}
  \flatalp \cont \cod\strone{\alpone^*}
  & = &
  \fix \flatauxalp \cont \cod\strone{\alpone^*}\\
  & \tobdet^2 &
  \flatauxalp (\la\varthree \fix \flatauxalp \varthree) \cont \cod\strone{\alpone^*}\\
  & = &
  \flatauxalp (\la\varthree\flatalp \varthree) \cont \cod\strone{\alpone^*}\\
  & = &
  (\la\var \la\cont \lambda \strone.\strone N_1\ldots N_{\size\alpone}N_{\elemblank}N\cont) (\la\varthree\flatalp \varthree) \cont \cod\strone{\alponeblank^*}\\
  & \tobdet &
  (\la\cont \lambda \strone.\strone N_1\ldots N_{\size\alpone}N_{\elemblank}N\cont) \isub\var{ (\la\varthree\flatalp \varthree) } \cont \cod\strone{\alponeblank^*}\\
  & = &
  (\lambda \cont.\lambda \strone.\strone P_1\ldots P_{\size\alpone}P_{\elemblank}N\cont)\cont\cod\strone{\alponeblank^*}\\
  & \tobdet^2 & 
  \cod\strone{\alponeblank^*}P_1\ldots P_{\size\alpone}P_{\elemblank} N\cont
\end{array}$$
Two cases, depending on the nature of $\strone$.
\begin{itemize}
\item \emph{$\strone$ is the empty string $\varepsilon$}. Then:
\begin{eqnarray*}
\cod{\varepsilon}{\alponeblank^*}P_1\ldots P_{\size\alpone}P_{\elemblank}N\cont
   & =& (\lambda\var_1.\ldots.\lambda\var_{\size\alpone}.\lambda\var_{\elemblank}.I) P_1\ldots P_{\size\alpone}P_{\elemblank}N\cont\\
   &\tobdet^{\size\alponeblank}& I N\cont\\
      &\tobdet& N\cont\\
   & = & (\la\vartwo \vartwo\cod{\varepsilon}{\alpone^*}) \cont \\
   &\tobdet &\cont\cod{\varepsilon}{\alpone^*}
\end{eqnarray*}

So we have $\flatalp \cont\cod\strone{\alponeblank^*} 
\tobdet^{n}
\cont\cod{\varepsilon}{\alpone^*}$ with $\size\strone = 0 \leq n = \size\alponeblank +6 = \funconvert{\alponeblank}(0) = \funconvert{\alponeblank}(\size\strone)$.

\item \emph{$\strone$ is a compound string $\elem_i\strtwo$}. Two sub-cases:
\begin{enumerate}
\item \emph{$\elem_i$ is an element of $\alpone$}: then the reasoning goes exactly has in the lifting case.

\item \emph{$\elem_i$ is $\elemblank$}: 
 \begin{eqnarray*}
  \cod{\elemblank\strtwo}{\alponeblank^*}P_1\ldots P_{\size\alpone}P_{\elemblank}N\cont
   &=& (\lambda\var_1.\ldots.\lambda\var_{\size\alpone}.\lambda\var_\elemblank.\lambda\vartwo.\var_\elemblank\cod{\strtwo}{\alponeblank^*}) P_1\ldots P_{\size\alpone}P_{\elemblank}N\cont\\
   &\tobdet^{\size\alponeblank}& 
   (\lambda\vartwo.P_\elemblank\cod{\strtwo}{\alponeblank^*}) N\cont\\
   &\tobdet& 
	N_\elemblank \cod{\strtwo}{\alponeblank^*} \cont
\end{eqnarray*}
by \ih there is $m$ such that $\size\strtwo \leq m \leq \funconvert{\alponeblank}(\size\strtwo)$ and

\begin{eqnarray*}
  P_\elemblank\cod\strtwo{\alptwo^*} \cont
  & = &
  (\la\strtwo \la\cont (\la\varthree\flatalp \varthree)  \cont\strtwo)\cod{\strtwo}{\alptwo^*}\cont\\
  & \tobdet^2 &
  \flatalp \cont \cod{\strtwo}{\alponeblank^*}\\
  (\mbox{\ih}) & \tobdet^m &
  \cont \cod{\strtwo}{\alpone^*}
\end{eqnarray*}
Summing up, $\flatalp \cont\cod\strone{\alptwo^*}
\tobdet^{n}
\cont\cod{\strone}{\alpone^*}$ with $n = m + \size\alponeblank + 8$. Note that $\funconvert{\alponeblank}(\size\strone) = \funconvert{\alponeblank}(\size\strtwo+1) = \funconvert{\alponeblank}(\size\strtwo) + \size\alptwo + c + 10$ and since $\size\strtwo \leq m \leq \funconvert{\alponeblank}(\size\strtwo)$ we obtain $\size\strone = \size\strtwo +1 \leq n \leq \funconvert{\alponeblank}(\size\strone)$.
\end{enumerate}
\end{itemize}
\end{proof}


\section{Encoding Turing Machines}

\paragraph{Turing Machines.}
A deterministic Turing machine $\mathcal{M}$ is a tuple $(\alponeblank,\States,\statein,
\statefin,\delta)$ consisting of:
\begin{itemize}
  \item
  A finite alphabet $\alpone=\{\elem_1,\ldots,\elem_n\}$ plus a distinguished symbol $\elemblank$, called the \emph{blank symbol};
  \item
  A finite set $Q=\{\state_1,\ldots,\state_m\}$ of \emph{states};
  \item
  A distinguished state $\statein\in Q$, called the \emph{initial
  state};
  \item
  A distinguished state $\statefin\in Q$, called the \emph{final
  state}; 
  \item
  A partial \emph{transition function} $\delta:Q\times\alpone\rightharpoonup 
  Q\times\alpone\times\{\leftarrow,\rightarrow,\downarrow\}$ such that
  $\delta(\state_i,\elem_j)$ is defined iff $\state_i\neq \statefin$.
\end{itemize}
A configuration for $\mathcal{M}$ is a quadruple
$(\strone, \elem, \strtwo, \state) \in \alpone^*\times\alpone\times\alpone^*\times Q$ where:
\begin{itemize}
\item $\strone$ is the tape on the left of the head;
\item $\elem$ is the element on the cell read by the head;
\item $\strtwo$ is the tape on the right of the head;
\item $\state$ is the state of the machine.
\end{itemize}
Given a string $\strone \in\alpone^*$ we define:
\begin{itemize}
  \item the \emph{initial configuration} $\initconfigs$
for $\strone$ is $\initconfigs \defeq (\varepsilon,\elemblank,\strone,\statein)$,
  \item the \emph{final configuration} $\finalconfigs$ for $\strone$ is $\finalconfigs \defeq (\varepsilon,\elemblank,\strone,\statefin)$
\end{itemize}

\emph{An example of transition}: if $\delta(\state_i,\elem_j) = (\state_l,\elem_k,\leftarrow)$, then $\M$ 
evolves from $\config = (\strone\elem_p,\elem_j,\strtwo,\state_i)$ to $\configtwo = (\strone,\elem_p,\elem_k\strtwo,\state_l)$ and if the tape on the left of the head is empty, \ie if $\config = (\varepsilon,\elem_j,\strtwo,\state_i)$, then the content of the new head cell is a blank symbol, that is $\configtwo \defeq (\varepsilon,\elemblank,\elem_k\strtwo,\state_l)$. If $\M$ has a transition from $\config$ to $\configtwo$ we write $\config \tomachtur \configtwo$. A configuration whose state is the final state $\statefin$ is \emph{final} and cannot evolve.

A Turing machine $(\alponeblank,\States,\statein,
\statefin,\delta)$ computes the function $f:\alpone^*\rightarrow\alpone^*$ 
 in time $g:\Nset\rightarrow\Nset$ 
iff for every $\strone\in\alpone^*$, the
initial configuration for $\strone$ evolves to a final configuration
for $f(\strone)$ in $g(\size\strone)$ steps.

\paragraph{Notation.} From now on, we fix a Turing machine $\M = (\alponeblank,\States,\statein,\statefin,\delta)$ and encode it. To ease the notation we remove the superscripts to the encoding of strings and elements, unless when necessary to disambiguate. The terms encoding configurations and those building the initial configuration, representing the transition function, and extracting the output are also relative to $\M$ but we avoid adding $\M$ as a superscript, to ease the notation.

\paragraph{Encoding configurations.} A configuration $(\strone,\elem,\strtwo,\state)$ of a machine
$\M = (\alponeblank,\States,\statein, \statefin,\delta)$ is represented by
the term
$$
\cod{(\strone,\elem,\strtwo,\state)}{\M}
\defeq
\lambda x. (x
\cod{\strone^r}{\alponeblank^*}\;\cod{\elem}{\alponeblank}\;\cod{\strtwo}{\alponeblank^*}\;\cod{\state}{\States}).
$$
where $\strone^r$ is the string $\strone$ with the elements in reverse order. We will often rather write
$$
\cods{(\strone,\elem,\strtwo,\state)}
\defeq
\lambda x. (x
\cods{\strone^r}\;\cods\elem \; \cods\strtwo \; \cods\state).
$$
letting the superscripts implicit. To ease the reading, we sometimes use the following  notation for tuples $\tuple{\tm,\tmtwo,\tmthree,\tmfour} \defeq \lambda x. (x \tm \tmtwo \tmthree \tmfour)$, so that $\cods{(\strone,\elem,\strtwo,\state)} = \tuple{\cods{\strone^r}, \cods\elem, \cods\strtwo, \cods\state}$.

\paragraph{Turning the input string into the initial configuration.} The following lemma provides the term $\inits $ that builds the initial configuration.

\begin{lemma}[Turning the input string into the initial configuration]
\label{l:init-config}
Let $\M=(\alponeblank,\States,\statein,
\statefin,\delta)$ be a Turing machine. There is a term $\init\M{}$, or simply $\inits$, such that for every $\strone\in\alpone^*$
$$\inits \cont \cod{\strone}{\alpone^*} \tobdet^{\Theta(\size\strone)} \cont \cods\initconfigs$$ 
where $\initconfigs$ is the initial configuration of $\M$ for $\strone$.
\end{lemma}

\begin{proof}
Define
\begin{align*}
\inits &\defeq \la\cont \liftalp(\la\strtwo \cont
\tuple{\cod{\varepsilon}{\alponeblank^*}, \cod{\elemblank}{\alponeblank}, \strtwo, \cod{\statein}{\States}}).
\end{align*}
Then 
$$\begin{array}{rclcl}
  \inits \cont \cod\strone{\alpone^*}
  & = &
  (\la\cont \liftalp(\la\strtwo \cont\tuple{\cod{\varepsilon}{\alponeblank^*}, \cod{\elemblank}{\alponeblank}, \strtwo, \cod{\statein}{\States}})) \cont \cod\strone{\alpone^*}\\
  & \tobdet &
  \liftalp(\lambda \strtwo.(\cont\tuple{\cod{\varepsilon}{\alponeblank^*}, \cod{\elemblank}{\alponeblank}, \strtwo, \cod{\statein}{\States}}))\cod{\strone}{\alpone^*}\\
  (\mbox{by \reflemmaeq{lifting-alphabets}}) & \tobdet^{\Theta(\size\strone)} &
  (\lambda \strtwo.(\cont\tuple{\cod{\varepsilon}{\alponeblank^*}, \cod{\elemblank}{\alponeblank}, \strtwo, \cod{\statein}{\States}}))\cod{\strone}{\alponeblank^*}\\
  & \tobdet &
  \cont \tuple{\cod{\varepsilon}{\alponeblank^*}, \cod{\elemblank}{\alponeblank}, \cod{\strone}{\alponeblank^*}, \cod{\statein}{\States}}\\
  & = &
  \cont \cods{(\varepsilon, \elemblank, \strone, \statein)}\\
  & = &
  \cont \cods\initconfigs
\end{array}$$
\end{proof}

\paragraph{Extracting the output from the final configuration.}

\begin{lemma}[Extracting the output from the final configuration]
\label{l:final-config}
Let $\M=(\alponeblank,\States,\statein,
\statefin,\delta)$ be a Turing machine. There is 
a term $\final\M{}$, or simply $\finals$, such that  for every final
configuration $\config$ for $\strone\in\alpone^*$

$$\finals \cont \cods\config \tobdet^{\Theta(\size\strone)} \cont \cod{\strone}{\alpone^*}.$$
\end{lemma}

\begin{proof}
Define
\begin{align*}
\finals &\defeq \la \cont \lambda y.y(\lambda v.\lambda a.\lambda \strone.\lambda q.\flatalp \cont \strone)
\end{align*}
Then:
$$\begin{array}{rclcl}
  \finals \cont \cods\config
  & = &
  (\la \cont \lambda y.y(\lambda v.\lambda a.\lambda \strone.\lambda q.\flatalp \cont \strone)) \cont \cods\config\\
  & \tobdet^2 &
  \cods\config(\lambda v.\lambda a.\lambda \strone.\lambda q.\flatalp\cont \strone)\\
  & = &
  \cods{ (\strthree, \elem, \strone, \statefin) }(\lambda v.\lambda a.\lambda \strone.\lambda q.\flatalp \cont \strone)\\
  & = &
  (\lambda x.x\cod{\strthree^r}{\alponeblank^*}\;\cod{\elem}{\alponeblank}\;\cod{\strone}{\alponeblank^*}\;\cod{\statefin}{\States}) (\lambda v.\lambda a.\lambda \strone.\lambda q.\flatalp\cont \strone)\\
  & \tobdet &
  (\lambda v.\lambda a.\lambda \strone.\lambda q.\flatalp\cont \strone) \cod{\strthree^r}{\alponeblank^*}\;\cod{\elem}{\alponeblank}\;\cod{\strone}{\alponeblank^*}\;\cod{\statefin}{\States}\\ 
  &\tobdet^4 &
  \flatalp\cont\cod{\strone}{\alponeblank^*}\\
  (\mbox{by \reflemmaeq{flattening-alphabets}}) & \tobdet^{\Theta(\size\strone)} &
  \cont\cod{\strone}{\alpone^*}
\end{array}$$
\end{proof}

\paragraph{Simulation of a machine transition.} Now we show how to encode the transition function $\delta$ of a Turing machine as a $\l$-term in such a way to simulate every single transition in constant time. This is the heart of the encoding, and the most involved proof.

\begin{lemma}[Simulation of a machine transition]
\label{l:trans-sim}
Let $\M=(\alponeblank,\States,\statein,
\statefin,\delta)$ be a Turing machine. There is a term $\trans{\M}$, or simply $\transs$, 
such that for every configuration $\config$ 
\begin{itemize}
  \item \emph{Final configuration}:
  if $\config$ is a final configuration
  then $\transs \cont \cods\config \tobdet^{O(1)} \cont \cods\config$;
  \item \emph{Non-final configuration}: If $\config \tomachtur \configtwo$ then $\transs \cont \cods\config \tobdet^{O(1)} \transs \cont \cods\configtwo$.
\end{itemize}
\end{lemma}

\begin{proof}
Define
\[\begin{array}{rclll}
\transaux 
& \defeq & 
(\lambda x.\la\cont\lambda y.y(\lambda u.\lambda a.\lambda v.\lambda q.q M_1\ldots M_{\size\States} au\cont v)),\\
\transs 
& \defeq &
\fix \transaux,
\end{array}\]
where, for any $i$ and $j$:
\begin{eqnarray*}
M_i&\defeq&\lambda a.aN_i^1\ldots N_i^{\size\alpone};\\
N_i^j&\defeq&
    \left\{
   \begin{array}{ll}
   \lambda u.\lambda \cont .\lambda v.\cont \tuple{u,\cods{\elem_j}, v,\cods{\state_i} }
      \hfill\mbox{if $\state_i=\statefin$}\\
      
   \lambda u.\lambda \cont .\lambda v. x\cont \tuple{ u,\cods{\elem_h}, v,\cods{\state_l}} 
      \hfill\mbox{\qquad if $\delta(\state_i,\elem_j)=(\state_l,\elem_h,\downarrow)$}\\
      
   \lambda u.uP_1^{l,h}\ldots P_{\size\alpone}^{l,h}P^{l,h} 
      \hfill \mbox{if $\delta(\state_i,\elem_j)=(\state_l,\elem_h,\leftarrow)$}\\ 
      
   \lambda u.\lambda v.vR_1^{l,h}\ldots R_{\size\alpone}^{l,h}R^{l,h}u 
      \hfill \mbox{if $\delta(\state_i,\elem_j)=(\state_l,\elem_h,\rightarrow)$;}
   \end{array}
    \right.\\
    P_i^{l,h}
      & \defeq &
      \lambda u.\lambda \cont .\appendchar{\elem_h}(\lambda w.x\cont 
\tuple{u,\cods{\elem_i},w, \cods{\state_l}})\\
    P^{l,h}
      & \defeq &
      \lambda \cont .\appendchar{\elem_h}(\lambda w.x\cont \tuple{\cods{\varepsilon},\cods{\elemblank},w, \cods{\state_l}})\\
    R_i^{l,h}
      & \defeq &
      \lambda u.\lambda \cont .\appendchar{\elem_h}(\lambda w.x\cont 
      \tuple{w, \cods{\elem_i}, u, \cods{\state_l}})\\
    R^{l,h}
      & \defeq &
      \lambda \cont .\appendchar{\elem_h}(\lambda w.x\cont 
      	\tuple{ w, \cods\elemblank, \cods\varepsilon, \cods{\state_l} }
      )
\end{eqnarray*}

Let $C=(\strone,\elem_j,\strtwo,\state_i)$
  First, we need some
 abbreviations. For every $i$ and $j$ define:
 \begin{eqnarray*}
\States_i&\defeq& M_i \isub\var{ \lambda z.\transs z } ;\\ 
 T_i^j&\defeq& N_i^j \isub\var{ \lambda z.\transs z }.
 \end{eqnarray*}

Then:
\begin{center}
$\begin{array}{llllll}
\transs  \cont  \cods\config & = & \fix \transaux \cont  \cods\config \\
& \tobdet &
\transaux (\la\varthree \fix \transaux \varthree) \cont  \cods\config\\
& = &
\transaux (\la\varthree \transs  \varthree) \cont  \cods\config\\
& = &
(\lambda x.\la\cont\lambda y.y(\lambda u.\lambda a.\lambda v.\lambda q.qM_1\ldots M_{\size\States}au\cont v)) (\la\varthree \transs  \varthree) \cont  \cods\config\\
& \tobdet^3 &
\cods\config (\lambda u.\lambda a.\lambda v.\lambda q.((qM_1\ldots M_{\size\States}) \isub\var{ \la\varthree \transs  \varthree } au \cont v))\\
& = &
\cods\config (\lambda u.\lambda a.\lambda v.\lambda q.q\States_1\ldots \States_{\size\States} au \cont v)\\
& = &

\cods{(\strone,\elem_j,\strtwo,\state_i)} (\lambda u.\lambda a.\lambda v.\lambda q.q\States_1\ldots \States_{\size\States} au \cont v)\\
& = &
(\lambda x. (x
\cods{\strone^r} \;\cods{\elem_j} \;\cods\strtwo \;\cods{\state_i} )) (\lambda u.\lambda a.\lambda v.\lambda q.q\States_1\ldots \States_{\size\States} au \cont v)\\
& \tobdet &
(\lambda u.\lambda a.\lambda v.\lambda q.q\States_1\ldots \States_{\size\States} au \cont v) \cods{\strone^r} \;\cods{\elem_j} \;\cods\strtwo \;\cods{\state_i} \\
& \tobdet &
\cods{\state_i}  \States_1\ldots \States_{\size\States}    \cods{\elem_j} \cods{\strone^r}   \cont \cods\strtwo  \\
& = &
(\lambda\var_1.\ldots.\lambda\var_{\size\States}.\var_i) \States_1\ldots \States_{\size\States}    \cods{\elem_j} \cods{\strone^r}   \cont \cods\strtwo \\
& \tobdet^{\size\States} &
\States_i    \cods{\elem_j} \cods{\strone^r}   \cont \cods\strtwo \\
& = &
(\lambda a.aT_i^1\ldots T_i^{\size\alpone}u)   \cods{\elem_j}  \cods{\strone^r}  \cont \cods\strtwo\\

& \tobdet &
\cods{\elem_j}  T_i^1\ldots T_i^{\size\alpone} \cods{\strone^r}    \cont  \cods\strtwo\\

& = &
(\lambda\var_1.\ldots.\lambda\var_{\size\alpone}.\var_j) T_i^1\ldots T_i^{\size\alpone} \cods{\strone^r}    \cont  \cods\strtwo\\

& \tobdet^{\size\alpone} &
T_i^j
\cods{\strone^r}    \cont \cods\strtwo
\end{array}$
 \end{center}
 
 Now, consider the following 
 four cases, depending on the value of $\delta(\state_i,\elem_j)$:
 \begin{enumerate}
    \item \emph{Final state}: 
    if $\delta(\state_i,\elem_j)$ is undefined, then $\state_i = \statefin$ and $T_i^j \defeq    \lambda u.\la\cont \lambda v.\cont \tuple{u,\cods{\elem_j}, v,\cods{\state_i} }
$, by definition.
Then:
\begin{center}
$\begin{array}{llllll}
    T_i^j \cods{\strone^r}    \cont \cods\strtwo
	& = &
    (   \lambda u.\la\cont \lambda v.\cont\tuple{u,\cods{\elem_j}, v,\cods{\state_i} }
) \cods{\strone^r}    \cont \cods\strtwo\\
    & \tobdet^3 &
   \cont  \tuple{\cods{\strone^r}, \cods{\elem_j}, \cods\strtwo, \cods{\state_i}}\\
    &=&\cont  \cods{(\strone ,\elem_j,\strtwo ,\state_i)}\\
    &=&\cont  \cods\config\\
\end{array}$
\end{center}

    \item \emph{The head does not move}: 
    If $\delta(\state_i,\elem_j)=(\state_l,\elem_h,\downarrow)$, then $\configtwo = (\strone, \elem_h,\strtwo,\state_l)$ and
\begin{center}
$\begin{array}{llllll}
    T_i^j & \defeq & N_i^j \isub\var{ \lambda z.\transs z }\\
    & = &
    (\lambda u.\la\cont \lambda v.x\cont \tuple{u,\cods{\elem_h}, v\cods{\state_l}})\isub\var{ \lambda z.\transs z }\\
    & = &\lambda u.\la\cont \lambda v. (\lambda z.\transs z)\cont \tuple{u,\cods{\elem_h}, v\cods{\state_l}}
\end{array}$
\end{center}

    Then:
    \begin{center}
$\begin{array}{llllll}
    T_i^j\cods{\strone^r} \cods\strtwo  \cont 
    	& = & (\lambda u.\la\cont \lambda v.(\lambda z.\transs z)\cont \tuple{u,\cods{\elem_h}, v,\cods{\state_l}}) \cods{\strone^r}   \cont \cods\strtwo\\
    &\tobdet^3&(\lambda z.\transs z)\cont \tuple{\cods{\strone^r},  \cods{\elem_h},  \cods\strtwo,  \cods{\state_l}}\\
    &\tobdet&\transs \cont \tuple{ \cods{\strone^r},  \cods{\elem_h},  \cods\strtwo,\cods{\state_l}} \\
    & = &
    \transs \cont \cods{(\strone,  \elem_h,  \strtwo , \state_l)} \\
    & = &\transs \cont  \cods\configtwo
\end{array}$
\end{center}

    \item \emph{The head moves left}: the hypothesis of this sub-cases is that $\delta(\state_i,\elem_j)=(\state_l,\elem_h,\leftarrow)$. So $T_i^j$ is going to depend on $P_i^{l,h}$ and $P^{l,h}$ defined at the beginning of the proof. Since $l$ and $h$ are now fixed we lighten the notation, using $P_i$ and $P$ for $P_i^{l,h}$ and $P^{l,h}$. Moreover, we define the following further abbreviations
    $$\begin{array}{lll}
      U_i & \defeq & P_i\isub\var{\la \varthree \transs  \varthree}\\
            & = & (\lambda u.\la\cont \appendchar{\elem_h} (\lambda w.x \cont \tuple{u,\cods{\elem_i},w, \cods{\state_l}}))\isub\var{\la \varthree \transs  \varthree}\\
      & = & \lambda u.\la\cont \appendchar{\elem_h} (\lambda w.(\la \varthree \transs  \varthree) \cont \tuple{u,\cods{\elem_i},w, \cods{\state_l}} )\\\\

      U & \defeq & P\isub\var{\la \varthree \transs  \varthree}\\
      & = &  (\la\cont \appendchar{\elem_h} (\lambda w.x \cont \tuple{\cods{\varepsilon}, \cods{\elemblank}, w,\cods{\state_l}} ))\isub\var{\la \varthree \transs  \varthree}\\

      & = &  \la\cont \appendchar{\elem_h} (\lambda w.(\la \varthree \transs  \varthree) \cont \tuple{\cods{\varepsilon}, \cods{\elemblank}, w,\cods{\state_l}})
    \end{array}$$
    With these conventions we have
    $$\begin{array}{lll}
      T_i^j & = & N_i^j \isub\var{\la \varthree \transs  \varthree} \\
      & = &
      (\lambda u. uP_1\ldots P_{\size\alpone}P)
                                                \isub\var{\la \varthree \transs  \varthree}\\
      & = &
      \lambda u. uU_1\ldots U_{\size\alpone}U
    \end{array}$$
    Then
    $$\begin{array}{lll}
    T_i^j \cods{\strone^r}   \cont \cods\strtwo
    & = & 
    (\lambda u.uU_1\ldots U_{\size\alpone}U) \cods{\strone^r}   \cont \cods\strtwo\\
    & \tobdet &
    \cods{\strone^r}  U_1\ldots U_{\size\alpone}U   \cont \cods\strtwo\\
    \end{array}$$
    
    Two sub-cases, depending on whether $\strone$ is an empty or a compound string.
    
    \begin{enumerate}
	 \item \emph{$\strone$ is the empty string $\varepsilon$}. Then $\configtwo = (\epsilon, \elemblank, \elem_h\strtwo, \state_l)$ and  
    $$\begin{array}{lll}
    \cods\epsilon U_1\ldots U_{\size\alpone}U   \cont \cods\strtwo
    & = & \\
    (\lambda\var_1.\ldots.\lambda\var_{\size\alpone}.\lambda\vartwo.\vartwo) U_1\ldots U_{\size\alpone}U   \cont \cods\strtwo
    & \tobdet^{\size\alpone +1} &\\
    U   \cont \cods\strtwo
    & = & \\
    (\la\cont \appendchar{\elem_h} (\lambda w.(\la \varthree \transs  \varthree) \cont \tuple{\cods{\varepsilon}, \cods{\elemblank}, w,\cods{\state_l}}))  \cont \cods\strtwo 
    & \tobdet & \\
    \appendchar{\elem_h} (\lambda w.(\la \varthree \transs  \varthree)\cont \tuple{\cods{\varepsilon}, \cods{\elemblank}, w,\cods{\state_l}} ) \cods\strtwo 
    & \tobdet^{ O(1) } & \reflemmaeq{append-char}\\
    (\lambda w.(\la \varthree \transs  \varthree)\cont \tuple{\cods{\varepsilon}, \cods{\elemblank}, w,\cods{\state_l}} ) \cods{(\elem_h\strtwo)}
    & \tobdet \\
    (\la \varthree \transs  \varthree)\cont \tuple{\cods{\varepsilon}, \cods{\elemblank}, \cods{(\elem_h\strtwo)}, \cods{\state_l} }
    & = \\
    (\la \varthree \transs  \varthree)\cont \cods{(\varepsilon, \elemblank, \elem_h\strtwo, \state_l)}
    & = \\
    (\la \varthree \transs  \varthree)\cont  \cods\configtwo
    & \tobdet \\
    \transs  \cont \cods\configtwo
    \end{array}$$

    \item \emph{$\strone$ is a compound string $\strthree\elem_p$}. Then $\configtwo = (\strthree, \elem_p, \elem_h\strtwo, \state_l)$ and 
    $$\begin{array}{lll}
    \cods{\strone^r}  U_1\ldots U_{\size\alpone}U   \cont \cods\strtwo
    & = & \\
    \cods{\elem_p\strthree^r} U_1\ldots U_{\size\alpone}U  \cont \cods\strtwo
    & = & \\    
    (\lambda\var_1.\ldots.
   \lambda\var_{\size\alpone}.\lambda\vartwo.\var_p\cods{\strthree^r} ) U_1\ldots U_{\size\alpone}U  \cont \cods\strtwo
    & \tobdet^{\size\alpone+1} & \\    
    U_p\cods{\strthree^r}  \cont \cods\strtwo
    & = & \\    
(\lambda u.\la\cont \appendchar{\elem_h} (\lambda w.(\la \varthree \transs  \varthree) \cont 
\tuple{u,\cods{\elem_p},w, \cods{\state_l}}
)) \cods{\strthree^r}   \cont \cods\strtwo
    & \tobdet^2 & \\ 
\appendchar{\elem_h} (\lambda w.(\la \varthree \transs  \varthree)\cont \tuple{ \cods{\strthree^r}, \cods{\elem_p},w,\cods{\state_l}} ) \cods\strtwo  
 & \tobdet^{ O(1) } & \reflemmaeq{append-char}\\
(\lambda w.(\la \varthree \transs  \varthree)\cont \tuple{\cods{\strthree^r}, \cods{\elem_p},w,\cods{\state_l}}) \cods{(\elem_h\strtwo)}
 & \tobdet\\
(\la \varthree \transs  \varthree)\cont (\tuple{\cods{\strthree^r}, \cods{\elem_p}, \cods{(\elem_h\strtwo)},  \cods{\state_l}}
 & =\\
(\la \varthree \transs  \varthree)\cont \cods{(\strthree, \elem_p, \elem_h\strtwo, \state_l)}
 & =\\

 (\la \varthree \transs  \varthree)\cont  \cods\configtwo
 & \tobdet\\
\transs  \cont  \cods\configtwo
    \end{array}$$
\end{enumerate}
    
    \item \emph{The head moves right}: the hypothesis of this sub-cases is that $\delta(\state_i,\elem_j)=(\state_l,\elem_h,\rightarrow)$. So $T_i^j$ is going to depend on $R_i^{l,h}$ and $R^{l,h}$ defined at the beginning of the proof. Since $l$ and $h$ are now fixed we lighten the notation, using $R_i$ and $R$ for $R_i^{l,h}$ and $R^{l,h}$. Moreover, we define the following further abbreviations
$$\begin{array}{lll}
      V_i & \defeq & R_i\isub\var{\la \varthree \transs  \varthree}\\
            & = & (\lambda u.\lambda \cont .\lambda v.\appendchar{\elem_h}(\lambda w.x\cont 
      \tuple{w, \cods{\elem_i}, u, \cods{\state_l}}) v) \isub\var{\la \varthree \transs  \varthree}\\
      & = & \lambda u.\lambda \cont .\appendchar{\elem_h}(\lambda w.(\la \varthree \transs  \varthree)\cont 
      \tuple{w, \cods{\elem_i}, u, \cods{\state_l}})\\\\

      V & \defeq & R\isub\var{\la \varthree \transs  \varthree}\\
      & = &  (\lambda \cont .\lambda v.\appendchar{\elem_h}(\lambda w.x\cont 
      	\tuple{ w, \cods\elemblank, \cods\varepsilon, \cods{\state_l} }
      ) v)\isub\var{\la \varthree \transs  \varthree}\\

      & = &  \lambda \cont .\appendchar{\elem_h}(\lambda w. (\la \varthree \transs  \varthree) \cont 
      	\tuple{ w, \cods\elemblank, \cods\varepsilon, \cods{\state_l} }
      )
    \end{array}$$
    
    With these conventions we have
    $$\begin{array}{lll}
      T_i^j & = & N_i^j \isub\var{\la \varthree \transs  \varthree} \\
      & = &
      (\lambda u.\la\cont \lambda v. vR_1\ldots R_{\size\alpone}R\cont u)
                                                \isub\var{\la \varthree \transs  \varthree}\\
      & = &
      \lambda u.\la\cont\lambda v. vV_1\ldots V_{\size\alpone}V\cont u
    \end{array}$$
    Then
    $$\begin{array}{lll}
    T_i^j \cods{\strone^r}   \cont \cods\strtwo
    & = & 
    (\lambda u.\la\cont\lambda v.vV_1\ldots V_{\size\alpone}U\cont u) \cods{\strone^r}   \cont \cods\strtwo\\
    & \tobdet^3 &
    \cods{\strtwo}  V_1\ldots V_{\size\alpone}V   \cont \cods{\strone^r}\\
    \end{array}$$
    
       Two sub-cases, depending on whether $\strtwo$ is an empty or a compound string.
    
    \begin{enumerate}
	 \item \emph{$\strtwo$ is the empty string $\varepsilon$}. Then $\configtwo = (\elem_h, \elemblank, \epsilon, \state_l)$ and  
        $$\begin{array}{lll}
    \cods\epsilon V_1\ldots V_{\size\alpone}V   \cont \cods{\strone^r}
    & = & \\
    (\lambda\var_1.\ldots.\lambda\var_{\size\alpone}.\lambda\vartwo.\vartwo) V_1\ldots V_{\size\alpone}V   \cont \cods{\strone^r}
    & \tobdet^{\size\alpone +1} &\\
    V   \cont \cods{\strone^r}
    & = & \\
    (\lambda \cont .\appendchar{\elem_h}(\lambda w. (\la \varthree \transs  \varthree) \cont 
      	\tuple{ w, \cods\elemblank, \cods\varepsilon, \cods{\state_l} }
      ))   \cont \cods{\strone^r}
    & \tobdet & \\
    \appendchar{\elem_h} (\lambda w.(\la \varthree \transs  \varthree)\cont \tuple{ w, \cods\elemblank, \cods\varepsilon, \cods{\state_l} })  \cods{\strone^r} 
    & \tobdet^{ O(1) } & \reflemmaeq{append-char}\\
    (\lambda w.(\la \varthree \transs  \varthree)\cont \tuple{ w, \cods\elemblank, \cods\varepsilon, \cods{\state_l} }) \cods{(\elem_h\strone^r)}
    & \tobdet \\
    (\la \varthree \transs  \varthree)\cont \tuple{ \cods{(\elem_h\strone^r)}, \cods\elemblank, \cods\varepsilon, \cods{\state_l} }
    & \tobdet \\
    \transs  \cont \tuple{ \cods{(\elem_h\strone^r)}, \cods\elemblank, \cods\varepsilon, \cods{\state_l} }
    & = \\
    \transs  \cont \cods{ (\strone\elem_h, \elemblank, \varepsilon, \state_l) }

    & = \\
    \transs  \cont \cods\configtwo
    \end{array}$$

    \item \emph{$\strtwo$ is a compound string $\elem_p\strthree$}. Then $\configtwo = (\strone\elem_h, \elem_p, \strthree, \state_l)$ and 
       $$\begin{array}{lll}
    \cods{\strtwo}  V_1\ldots V_{\size\alpone}V   \cont \cods{\strone^r}
    & = & \\
    \cods{\elem_p\strthree} V_1\ldots V_{\size\alpone}V   \cont \cods{\strone^r}
    & = & \\    
    (\lambda\var_1.\ldots.
   \lambda\var_{\size\alpone}.\lambda\vartwo.\var_p\cods{\strthree} ) V_1\ldots V_{\size\alpone}V   \cont \cods{\strone^r}
    & \tobdet^{\size\alpone+1} & \\    
    V_p\cods{\strthree}   \cont \cods{\strone^r}
    & = & \\    
(\lambda u.\lambda \cont .\appendchar{\elem_h}(\lambda w.(\la \varthree \transs  \varthree)\cont 
      \tuple{w, \cods{\elem_p}, u, \cods{\state_l}})) \cods{\strthree}   \cont \cods{\strone^r}
    & \tobdet^2 & \\ 
\appendchar{\elem_h}(\lambda w.(\la \varthree \transs  \varthree)\cont 
      \tuple{w, \cods{\elem_p}, \cods{\strthree}, \cods{\state_l}})\cods{\strone^r}
 & \tobdet^{ O(1) } & \reflemmaeq{append-char}\\
(\lambda w.(\la \varthree \transs  \varthree)\cont 
      \tuple{w, \cods{\elem_p}, \cods{\strthree}, \cods{\state_l}}) \cods{(\elem_h\strone^r)}
 & \tobdet\\
(\la \varthree \transs  \varthree)\cont 
      \tuple{\cods{(\elem_h\strone^r)}, \cods{\elem_p}, \cods{\strthree}, \cods{\state_l}}
       & \tobdet\\
      \transs \cont \tuple{\cods{(\elem_h\strone^r)}, \cods{\elem_p}, \cods{\strthree}, \cods{\state_l}}
 & =\\

\transs \cont \cods{(\strone\elem_h, \elem_p, \strthree, \state_l)}
 & =\\
\transs  \cont  \cods\configtwo
    \end{array}$$

\end{enumerate}

     \end{enumerate}
 
 \end{proof}

Straightforward inductions on the length of executions provide the following corollaries.
 
\begin{corollary}[Executions]
\label{coro:exec}

Let $\M$ be a Turing machine. The term $\transs$ encoding $\M$ as given by \reflemma{trans-sim} is such that
for every configuration $\config$ 
\begin{enumerate}
  \item \emph{Finite computation}:
  if $\configtwo$ is a final configuration reachable from $\config$ in $n$ transition steps
  then $\transs \cont \cods\config \tobdet^{O(n)} \cont \cods\configtwo$;
  \item \emph{Diverging computation}: if there is
  no final configuration reachable from $\config$ then $\transs \cont \cods\config$ diverges.
\end{enumerate}
\end{corollary}

 \paragraph{The simulation theorem.}
We now have all the ingredients for the final theorem of this note.

\begin{theorem}[Linear simulation]
\label{thm:main-simulation}
Let $\alpone$ be an alphabet and $f:\alpone^*\rightarrow\alpone^*$ a function computed by a Turing machine
$\M$ in time $g$. Then there is an encoding $\cods{\cdot}$ into $\detLam$ of $\alpone$, strings, and Turing machines over $\alpone$ such that for every $\strone\in\alpone^*$,
$\cods\M \cods\strone \tobdet^n \cods{f(\strone)}$
where $n=\Theta(g(\size\strone)+\size\strone)$.
\end{theorem}

\begin{proof}
Morally, the term is simply
$$
\function \defeq \inits (\transs(\finals(\lambda w.w))
$$
where $\lambda w.w$ plays the role of the initial continuation.

Such a term however does not belong to the deterministic $\l$-calculus, because the right subterms of applications are not always values. The solution is simple, it is enough to $\eta$-expand the arguments. Thus, define

$$
\function \defeq \inits (\la\vartwo\transs(\la \var\finals(\lambda w.w)\var)\vartwo)
$$

Then
$$\begin{array}{lll}
    \function \cods{\strone} 
    & = & \\
    \inits (\la\vartwo\transs (\la \var\finals (\lambda w.w)\var)\vartwo) \cods{\strone} 
    & \tobdet^{\Theta(\size\strone)} & (\mbox{by }\reflemmaeq{init-config})\\    
    (\la\vartwo\transs (\la \var\finals (\lambda w.w)\var)\vartwo) \cods{\initconfig}
    & \tobdet\\
    \transs (\la \var\finals (\lambda w.w)\var) \cods{\initconfig}
    & \tobdet^{\Theta(\size{g(\strone)})} & (\mbox{by }\refcoroeq{exec})\\
    (\la \var\finals (\lambda w.w)\var) \cods{\config_{\tt fin}(f(\strone))}
    & \tobdet\\
    \finals (\lambda w.w) \cods{\config_{\tt fin}(f(\strone))}
    & \tobdet^{\Theta(\size\strone)} & (\mbox{by }\reflemmaeq{final-config})\\    
    (\lambda w.w) \cods{f(\strone)}
    & \tobdet &\\
    \cods{f(\strone)}
    \end{array}$$
\end{proof}

\bibliographystyle{alpha}

\begin{thebibliography}{ADL12}

\bibitem[Acc17]{LSFA2017invited}
Beniamino Accattoli.
\newblock ({I}n){E}fficiency and {R}easonable {C}ost {M}odels.
\newblock In {\em {LSFA} 2017}, 2017.

\bibitem[ADL12]{DBLP:conf/rta/AccattoliL12}
Beniamino Accattoli and Ugo Dal~Lago.
\newblock On the invariance of the unitary cost model for head reduction.
\newblock In {\em RTA}, pages 22--37, 2012.

\end{thebibliography}


\end{document}